\documentclass{article}
\usepackage{amsmath,amsthm,amsfonts,graphicx}

\usepackage{hyperref}
\newtheorem{theorem}{Theorem}
\newtheorem{corollary}[theorem]{Corollary}
\newtheorem{lemma}[theorem]{Lemma}
\theoremstyle{definition}
\newtheorem{definition}[theorem]{Definition}
\newtheorem{algorithm}[theorem]{Algorithm}

\newtheorem{remark}[theorem]{Remark}

\usepackage{xspace}
\newcommand{\sO}{\ensuremath{\widetilde{{O}}}\xspace}
\newcommand{\OB}{\ensuremath{{O}_B}\xspace}
\newcommand{\sOB}{\ensuremath{\widetilde{{O}}_B}\xspace}
\def\abs#1{\mathopen| #1 \mathclose|}	
\def\norm#1{\mathopen\| #1 \mathclose\|}

\def\wt#1{\widetilde #1}

\def\normi#1{\norm{#1}_{\infty}}

\begin{document}
\title{Accelerated Approximation of the Complex  Roots 
and Factors of a Univariate Polynomial}

\author{Victor Y. Pan$^{[1,2,a]}$, Elias P. Tsigaridas$^{[3]}$, Vitaly Zaderman$^{[2,b]}$,  and  
Liang Zhao$^{[2,c]}$\\[5pt]
\and\\
$^{[1]}$ Department of Mathematics and Computer Science \\
Lehman College of the City University of New York \\
Bronx, NY 10468 USA \\
and\\
$^{[2]}$Ph.D. Programs in Mathematics  and Computer Science \\
The Graduate Center of the City University of New York \\
New York, NY 10036 USA \\
\and\\
 $^{[a]}$victor.pan@lehman.cuny.edu \\
http://comet.lehman.cuny.edu/vpan/  \\
\and\\
$^{[2]}$  INRIA, Paris-Rocquencourt Center, {\em PolSys} \\
    Sorbonne Universit\'es, UPMC Univ. Paris 06, {\em PolSys},  \\
UMR 7606, LIP6, F-75005, Paris, France  \\
elias.tsigaridas@inria.fr \\
\and\\
 $^{[b]}$
vzaderman@yahoo.com, and
 $^{[c]}$lzhao1@gc.cuny.edu
 \\
\and\\
The preliminary version of this paper has been presented at SNC 2014}

 \date{}

\maketitle

\begin{abstract}
The algorithms of Pan (1995) \cite{P95} and Pan (2002) \cite{P02})
approximate the roots
 of a complex univariate  polynomial
in  nearly optimal arithmetic and Boolean time but 
require a  precision of computing that exceeds the degree of 
the polynomial. This causes numerical stability problems
 when the degree is large. 
We observe, however, that such a difficulty disappears
 at the initial stage of the algorithms,
and in our present paper we extend this stage 
to root-finding within a
nearly optimal arithmetic and Boolean
complexity bounds provided that some mild initial isolation of the roots 
of the input polynomial has been ensured. Furthermore our algorithm
is nearly optimal
for the approximation of the roots isolated in a fixed disc, square 
or another region on the complex plane rather than all complex 
roots of a polynomial. Moreover
the algorithm can be applied to a polynomial
given by a black box for its evaluation (even if its coefficients
are not known); it
 promises to be of practical value for polynomial root-finding
and  factorization,
the latter task being of interest on its own right. 
We also provide a new support for a winding number algorithm,
which enables extension of our progress to obtaining mild initial approximations
to the roots.
We conclude
  with summarizing our algorithms and their extension to  the
  approximation of isolated multiple roots and root clusters.
\end{abstract}





\section{Introduction}\label{sintr}
 



The classical problem of univariate polynomial root-finding has been central 
in Mathematics and Computational Mathematics for about four millennia since 
the Sumerian times, and is still important for Signal and Image Processing,
Control, Geometric Modeling, Computer Algebra, and Financial Mathematics.
It is closely linked to the approximation of linear and nonlinear factors
of a polynomial, which is also important on its own right
because of the applications to the time series analysis, Weiner filtering, noise 
variance estimation, covariance matrix computation, and
the study of multi-channel systems
(see Wilson (1969) \cite{W69},  Box and Jenkins
(1976)  \cite{BJ76}, Barnett (1983)  \cite{B83}, 
Demeure and  Mullis (1989 and 1990)  \cite{DM89}, \cite{DM90}, 
 Van Dooren (1994)) \cite{VD94}.

 Solution of both problems within nearly optimal 
 arithmetic and Boolean complexity bounds
(up to polylogarithmic factors)
have been obtained 
 in Pan (1995) \cite{P95} and Pan (2002) \cite{P02}, but
the supporting algorithms require a  precision of computing that exceeds
 the degree of 
the input polynomial, and this causes numerical stability problems
 when the degree is large.

The most popular packages of numerical subroutines
for complex polynomial root-finding, such as MPSolve 2000
(see Bini and Fiorentino (2000) \cite{BF00}), EigenSolve 2001 (see Fortune (2002)
\cite{F02}), and MPSolve 2012 (see Bini and Robol (2014) \cite{BR14})
employ alternative root-finders based on functional iterations 
(namely, Ehrlich--Aberth's and WDK, that is,
Weierstrass', also known as Durand-Kerner's) and the QR algorithm applied to 
eigen-solving for the companion matrix of the input polynomial.
The user considers these root-finders practically superior 
by relying on the empirical data about their excellent convergence,
even though these data have no formal support.
To their disadvantage, these algorithms compute 
the roots of a polynomial in an isolated region of the complex plane 
not much faster than all its roots.

We re-examine the subject, still assuming input polynomials with complex coefficients, and show that
 the cited deficiency of the algorithms of
\cite{P95} and \cite{P02} disappears if we 
 modify the initial 
 stage of these algorithms and apply them
 under some mild assumptions
about the initial isolation of the root sets of
the input polynomial. Moreover, like the algorithms of \cite{P95} and \cite{P02}
and unlike WDK and  Ehrlich--Aberth's algorithms,
the resulting algorithms are nearly optimal
for the approximation of the roots in an isolated region of the complex plane.

Next we briefly comment on our results.
In the next sections we elaborate upon them,
 deduce 
the computational cost estimates, and outline some natural extensions.




Recall that polynomial root-finding iterations can be partitioned into two stages.
 At first
 a crude (although reasonably good) initial
approximations to all roots or to a set of roots
are relatively slowly computed.
Then these approximations
are refined faster  
by means of the same or distinct
iterations. 

Our first algorithm applies at the second stage
and is nearly optimal, 
under both arithmetic and Boolean complexity models and
 under  mild initial isolation
of every root, some roots or some root sets. 
Such an isolation can be observed at some stages of root approximation
 by Ehrlich--Aberth's and WDK algorithms,
 but with no estimates for the computational cost of 
reaching isolation. 
Towards the solution with controlled computational cost, one can apply
advanced variants of
 Quad-tree construction  of Weyl 1924 \cite{W24}, successively refined 
in Henrici 1974 \cite{H74}, 
Renegar 1987 \cite{R87}, and Pan 2000 \cite{P00}. 

The algorithm of the latter paper computes all roots 
of a polynomial
or its roots in a fixed isolated region 
at a 
 nearly optimal arithmetic cost,
and it is nearly optimal for computing the initial
isolation of the roots as well; 
the paper \cite{P00} does not estimate the Boolean cost
of its algorithm, but most of its steps allow rather  straightforward
 control of the precision of computing.\footnote{More recent variations of the
Quad-tree algorithm have been studied by various authors under the name of subdivision
algorithms for polynomial root-finding.}

Moreover in Section \ref{smhr} we present a {\em new winding number algorthm} 
that computes the number of roots in a fixed square on a complex plane
at a nearly optimal Boolean cost. By incorporating 
this algorithm into the refined  variant of Weyl's Quad tree construction  of \cite{P00},
we   extend our progress to  
obtaining mild initial isolation of all
the roots of a polynomial as well as its roots isolated in a fixed region. 
This root-finder is performed at a nearly optimal Boolean cost
and can  be applied even where a polynomial is given by 
a black box for its evaluation.

If properly implemented our algorithms 
have very good chances to
become the user's choice.

Besides the root-finding applications,
they can be valuable ingredients 
of the polynomial factorization algorithms.
Recall that one can extend factorization 
 to root-finding (see Sch{\"o}nhage (1982)
\cite{schoenhage82},
\cite{P95}, \cite{P02},  and the present paper),
but also root-finding
to factorization (see
Pan (2012) \cite{P12}).
Our algorithm  can be technically linked to those of
\cite{schoenhage82},
\cite{P95}, \cite{P02}; our
results (Theorem~\ref{thm:cref-one} and \ref{thm:cref-all}) could be 
also viewed as an extension of the recent record and nearly optimal bounds for the approximation of the real roots
 \cite{PT13,PTa}, see also \cite{MNP12}.

 An interesting  
challenge is the design of
a polynomial factorization algorithms that both are 
simple enough for practical implementation and
support factorization at a nearly optimal computational
complexity.   
An efficient solution 
 outlined in our last section
combines our present algorithms with 
the one of
Pan 2012 \cite{P12}
and McNamee and Pan 2013
\cite[Section 15.23]{MNP13},
which is a simplified version of
the 
efficient but very much 
involved algorithm of Kirrinnis (1998) \cite{K98}.
 
Like our present algorithm,
these solution algorithms are  nearly optimal
and remain  nearly optimal when they are
 applied to a polynomial given by a black box for its
evaluation, even when its coefficients are not known.


\paragraph{ Organization of the paper.}
We recall the relevant definitions and
some basic results  
in the remainder of this section and in the next section.
In Section \ref{sisrrr} we present our main algorithm,
prove its correctness, and 
estimate its arithmetic cost 
when it is applied to the approximation of a single
root and $d$ simple isolated roots of a $d$th degree polynomial.
In Section \ref{sbcbr} we extend our analysis to estimate 
the Boolean cost of these computations.
In Section \ref{smhr} we present our new winding number algorthm.
In our concluding Section \ref{sconcl} 
we summarize our results and 
outline their extension to factorization of 
a polynomial and to root-finding in the cases  of
isolated 
multiple roots and root  clusters.

\paragraph{Some definitions.}

\begin{itemize}
\item
For  a polynomial $u=u(x)=\sum^{d}_{i=0}u_ix^i$, the norms
$||u||_{\gamma}$ denote the norms $||{\bf u}||_{\gamma}$ of its coefficient
vector ${\bf u}=(u_i)^{d}_{i=0}$, for $\gamma=1,2,\infty$. 
\item
$D(X, r)$  denotes the complex disc $\{x: |x-X|\le r\}$.
 \item
``ops" stands for ``arithmetic operations".
 \item
$DFT(q)$ denotes the discrete Fourier transform at $q$ points.
It can be performed by using $O(q\log(q))$ ops.
\end{itemize}


\section{Isolation Ratio and Root-refinement}\label{sirrr}
 
 
The following concept of the {\em isolation ratio} is
 basic for us, as well as for \cite{P95} and  \cite{P02}.
Assume a real or complex polynomial 
\begin{equation}\label{eqpoly}
p=p(x)=\sum^{d}_{i=0}p_ix^i=p_n\prod^d_{j=1}(x-z_j),~~~ p_d\ne 0,
\end{equation}
of degree $d$, an annulus $A(X,R,r)=\{x: r\le |x-X| \le R \}$
 on the complex plane
 with a center $X$, and
the radii $r$ and $R$ of the boundary circles. Then the 
 internal
disc $D(X, r)$ is  $R/r$-{\em isolated},
and we call $R/r$ its {\em isolation ratio} 
if  the polynomial $p$ has no roots 
in the annulus.  
The isolation ratios, for all discs $D(0,r)$ and for all positive $r$,
can be approximated as long as we can approximate the root radii $|z_j|$,
for $j=1,\dots,d$. 

The algorithms of \cite{schoenhage82}  
(cf. also  Pan (2000) \cite{P00} and \cite{P02}) yield such approximations within 
a constant relative error, say, 0.01, by using $O(d\log^2(d))$ ops,
but involve the Dandelin's root-squaring iteration (see Householder (1959) \cite{H59}),
and this leads to numerical stability problems. 

Alternative heuristic 
algorithms of Bini (1996) \cite{B96} and \cite{BF00} are slightly faster, but also
cannot produce close approximation without using 
root-squaring iteration. 

The Schur-Cohn test does not 
use these iterations and can be applied to 
estimate the isolation ratio more directly. 
For the disc $D(0,r)$,
a variant of this test 
in Renegar (1987) \cite[Section 7]{R87} amounts 
 to performing FFT at $d'=2^h$ points, for $16d\le d'\le 32d$,
with the overhead of $O(n)$ ops and comparisons of real numbers with 0.
This means a reasonably low precision of computing
and Boolean cost. 
 Also see Brunie and Picart (2000) \cite{BP00}.

The following result from Tilli (1998) \cite{T98}
shows that
 Newton's classical iteration converges quadratically
 to a  single simple root of $p$ if it is initiated at the center of a $3d$-isolated disc
that contains just this root. The result softens the restriction that  $s\ge 5d^2$ of
\cite[Corollary~4.5]{R87}.
\begin{theorem}\label{thren}
  Suppose that both discs $D(c, r)$ and $D(c, r/s)$, for $s\ge 3d$,
  contain a single simple root 
  $\alpha$ of 
  a polynomial $p=p(x)$ of degree $d$.
  Then  Newton's
  iteration 
  \begin{equation}\label{eqnewt}
    x_{k+1}=x_k-p(x_k)/p'(x_k),  k=0,1,\dots
  \end{equation}
  converges quadratically to the root $\alpha$
  right from the start provided that $x_0=c$.
\end{theorem}


\section{Increasing Crude Isolation Ratios of Polynomial Roots}\label{sisrrr}
 

We can shift and scale the variable $x$, and so 
with no loss of generality we assume dealing with 
a $(1+\eta)^2$-isolated disc $D(0,r)$, for $r=1/(1+\eta)$, a fixed $\eta>0$,
and a polynomial $p$ of (\ref{eqpoly}) having precisely $k$  
not necessarily distinct roots
 $z_1,\dots,z_k$ in this disc.

In Section \ref{sconcl} we outline some
important extensions of our current study based on our results 
in the general case of any $k<d$, which we produce next,
but in 
 Sections \ref{sisrrr} and \ref{sbcbr} we assume that $k=1$.

Now, under the above assumptions for any $k<d$,
can we increase the isolation  ratio, say, to $3d$?
 We apply the following 
algorithm.


\begin{algorithm}\label{algpwrs} {\bf The Power Sums.}


\begin{description}


\item[{\sc Input:}] 
three integers $d$, $k$ and $q$,
such that   $0<k<n$, $k<q$, and $\eta > 0$,

$\omega=\omega_q=\exp(2\pi\sqrt {-1}/q)$, a primitive $q$th
root of unity,

the coefficients $p_0,\dots,p_d$ of 
a  polynomial $p=p(x)$ of
(\ref{eqpoly}).

We write $r=1/(1+\eta)$ and assume that the   disc $D(0,r)$ 

(i) is $(1+\eta)^2$-isolated
and  \\
(ii)  contains the $k$ roots $z_1,\dots,z_k$ of 
the polynomial $p=p(x)$ and no other roots.


\item[{\sc Output:}] 
the values $\sigma_1^*,\dots,\sigma_k^*$ such that
\begin{equation}\label{eqsgm*}
|\sigma_g-\sigma_g^*|\le \Delta_{k,q}=(r^{q+k}+(d-1)r^{q-k})/(1-r^{q}),
~{\rm for}~g=1,\ldots,k,
\end{equation} 

\begin{equation}\label{eqpsm}
\sigma_g=\sum_{j=1}^kz_{j}^g,~{\rm for}~g=1,\ldots,k.
\end{equation}
 

\item{\sc Computations}: $~$ 


\begin{enumerate}


\item 
Compute 
  the coefficients 
of the two auxiliary
polynomials $p_q(x)=\sum_{i=0}^lp_{q,i} x^i$ and
 $\bar p_q(x)=\sum_{i=0}^l\bar  p_{q,i} x^i$
where 
$p_{q,i}=\sum_{j=0}^lp_{i+jq}$ 
and $\bar p_{q,i}=\sum_{j=0}^{\bar l}p_{i+1+jq}$, 
for $i=0,\dots,q-1$,
$l=\lfloor d/q\rfloor$ and
$\bar  l=\lfloor (d-1)/q\rfloor$.


\item 
Compute the values 
$p_{q}(\omega^j)$, $\bar p_{q}(\omega^j)$,  and
 $r_j=p_{q}(\omega^j)/\bar p_{q}(\omega^j)$,  for $j=0,1,\dots,q-1$,


\item 
Compute and output the values
$\sigma_g^*=\frac{1}{q}\sum_{j=0}^{q-1}\omega^{(g+1)j}r_j$, for
$g=1,\ldots,k$.


\end{enumerate}


\end{description}


\end{algorithm}


\begin{theorem}\label{thops}
Algorithm \ref{algpwrs} involves $3d+O(q\log (q))$ ops.
\end{theorem}
\begin{proof} 
Stage 1 of the 
algorithm involves $d-1$ multiplications and less than $2d$ additions,
Stage 2 amounts to performing two DFT$(q)$
and $q$ divisions, and Stage 3 amounts to performing an DFT$(q)$
 (because $\omega^{g+1}$, for
$g=1,\ldots,k$, is the set of all $q$th roots of 1)
and $k$ divisions.
\end{proof}

In order to prove bound (\ref{eqsgm*}), implying 
{\em correctness of the algorithm},
at first observe that $p(\omega^j)=p_{q}(\omega^j)$ and 
$p'(\omega^j)=\bar p_{q}(\omega^j)$, for 
 $j=0,1,\dots,q-1$, and hence
\begin{equation}\label{equ7.12.2}
\sigma_g^*=\frac{1}{q}\sum_{j=0}^{q-1}\omega^{(k+1)j}p(\omega^j)/p'(\omega^j),
~{\rm for}~g=1,\ldots,k.
\end{equation} 

The proof of bound (\ref{eqsgm*}) also exploits the
Laurent expansion

\begin{equation}\label{equ7.12.3}
p'(x)/p(x)=\sum_{j=1}^d\frac{1}{x-z_j}= -\sum_{g=1}^\infty \sigma'_g
x^{g-1}+\sum_{g=0}^\infty \sigma_gx^{-g-1}= \sum_{h=-\infty}^{\infty}c_h
x^h
\end{equation}
where $|x|=1$,
$\sigma_0=1,~~\sigma_g=\sum_{j=1}^kz_{j}^g$
(cf. (\ref{eqpsm})), $\sigma'_g=
\sum_{i=k+1}^d~z_{i}^{-g},~g=1,2,\ldots$,
that is,
 $\sigma'_g$ is the $g$th power sum of the roots of
the reverse polynomial $p_{\textrm{rev}}(x)$ that lie in the disc
$D(0,r)$. 
The leftmost equation of
(\ref{equ7.12.3}) is verified by the differentiation of
$p(x)=p_n\prod_{j=1}^d(x-z_j)$. The middle equation 
is implied by the  decompositions 
$\frac{1}{x-z_{1}}
=~\frac{1}{x}\sum_{h=0}^\infty\left(\frac{z_{1}}{x}\right)^h
$ and 
$\frac{1}{x-z_{i}}
=~-\frac{1}{z_{i}}\sum_{h=0}^\infty\left(\frac{x}{z_{i}}\right)^h,
\mbox{
for}~i>1$, provided that $|x|=1$ for all $i$.
(Note a link of these expressions with the following quadrature formulae
for numerical integration,
$\sigma_g=\frac{1}{2\pi\sqrt{-1}}\int_{C(0,1)} x^mp'(x)/p(x) dx$, for
$g=1,\dots,k.)$

In order to 
deduce bound (\ref{eqsgm*}), we next combine equations (\ref{equ7.12.2}) and
(\ref{equ7.12.3}) and obtain
$$\sigma_k^*=\sum_{l=-\infty}^{+\infty}c_{-k-1+lq}.$$ Moreover, equation
(\ref{equ7.12.3}), for $h=-k-1$ and $k\ge 1$, implies that $\sigma_k=c_{-k-1}$,
while the same equation, for $h=k-1$ and $k\ge 1$, implies that
$\sigma'_k=-c_{k-1}$. Consequently
$$\sigma_k^*-\sigma_k=\sum_{l=1}^{\infty}(c_{lq-k-1}+c_{-lq-k-1}).$$ We assumed
in (\ref{equ7.12.2}) that $0<k<q-1$. It follows that
$c_{-lq-k-1}=\sigma_{lq+k}$ and $c_{lq-k-1}=-\sigma'_{lq-k}$, for
$l=1,2,\ldots$, and we obtain
\begin{equation}\label{equ7.12.5}
\sigma_k^*-\sigma_k=\sum_{l=1}^\infty(\sigma_{lq+k}-\sigma'_{lq-k}).
\end{equation}
Now recall that $|\sigma_h|\le
z^h$ and $|\sigma'_h|\le(d-1)z^h$, for $h=1,2,\ldots$ and 
$z=\max_{j=1}^{d}\min(|z_j|,1/|z_j|)$,
and so $z\le \frac{1}{1+t}$ in our case.
Substitute these bounds into (\ref{equ7.12.5}) and obtain
$$|\sigma_k^*-\sigma_k|\le(z^{q+k}+(d-1)z^{q-k})/(1-z^q).$$
Therefore, bound (\ref{eqsgm*}) follows because $z\le r$.

By substituting $q$ of order $\log (d)$ into bound (\ref{eqsgm*}), we can 
increase the isolation ratio of the disc $D(0,r)$ 
 by a factor of $gd^h$, for any pair of positive constants $g$ and
$h$. Therefore we obtain the following estimates.


\begin{theorem}\label{thisol}
Suppose the disc $D(0,r)=\{x:~|x|\le r\}$ is $(1+\eta)^2$-isolated,
for $(1+\eta)r=1$ and a fixed $\eta>0$, and
 contains exactly $k$ roots of a polynomial $p=p(x)$ of degree $d$. 
Let $g$ and $h$ be a pair of positive constants. Then 
it is sufficient to perform 
 $O(3d+\log (d)\log (\log (d)))$ ops
in order to compute a $gd^h$-isolated subdisc of $D(0,r)$ containing exactly the
same roots of $p=p(x)$.
\end{theorem}

If $k=1$, then $\sigma_1=z_1$, and by combining Theorems \ref{thren} and \ref{thisol} we obtain the following result.


\begin{corollary}\label{corref}
Let a polynomial $p(x)$ satisfy the assumptions of Theorem \ref{thisol}, for $k=1$.
Then we can approximate its root within  $\epsilon$, for $0<\epsilon<1$, by using 
$O(\log (d)\log (\log (d))+d\log (\log (1/\epsilon)))$ ops.
\end{corollary}


Hereafter we write
\begin{equation}\label{eqz}
z_+=\max(|z_{1}|,|z_2|,\dots,|z_d|).
\end{equation}

\begin{corollary}\label{corref1}
Suppose that we are given $d$ discs, each containing a single simple root 
of a polynomial $p=p(x)$ of degree $d$ and each being $(1+\eta)^2$-isolated, 
for a fixed $\eta>0$.
Then we can approximate all $d$ roots  of this polynomial 
within $\epsilon z_+$, for $z_+$ of (\ref{eqz}) and a
fixed $\epsilon$, $0<\epsilon<1$, by using 
$O(d\log^2 (d)(1+\log (\log (1/\epsilon))))$ ops.
\end{corollary} 


\begin{proof}
 Apply Algorithm \ref{algpwrs} 
concurrently in all $d$ given discs, but instead of the
$q$th roots of unity use $q$ equally spaced points at the boundary 
circle of each input disc (that is, $dq=O(d\log d)$ points overall)
and instead of DFT$(q)$ apply 
the Moenck--Bo\-ro\-din algorithm for multipoint polynomial evaluation
\cite{MB72}.

Also use it at the stage of performing concurrent Newton's iteration
 initialized at the centers  of the $3d$-isolated subdiscs of the 
$d$ input discs, each subdisc computed by the
 algorithm that supports Theorem \ref{thisol}.
Here we work with the $d$th degree  polynomial $p$ rather than with
the $q$th degree polynomials $p_q$.
Indeed, in order to support transition to polynomials $p_q$ of degree $q$,
 for $d$ discs, we would need to perform $d$ shifts and scalings
of the variable $x$, which would involve the order of $d^2\log(d)$ ops,
whereas by employing the Moenck--Bo\-ro\-din algorithm,
we obtain a nearly optimal root-refiner, which involves $O(d)$
ops up to polylogarithmic factors.

We replace the   
matrix $\Omega=[\omega^{j(k+1)}]_{j,k}$ in (\ref{equ7.12.2})  
by the matrix $[c+\omega^{j(k+1)}]_{j,k}=c[1]_{j,k}+\Omega$ where $c$ is 
 invariant in $j$ and $k$.
The multiplication of the new matrix by a vector ${\bf v}$ is still reduced to 
multiplication of  the matrix $\Omega$ by a vector ${\bf v}$
and to additional $3d$ ops, for computing the vector $c[1]_{j,k}{\bf v}$
and adding it to the vector $\Omega {\bf v}$.
\end{proof}

The Moenck--Borodin algorithm uses nearly linear arithmetic time, and
\cite{K98}
proved that this algorithm supports  multipoint 
polynomial evaluation at a low Boolean cost as well 
(see also
J. van der Hoeven (2008)
 \cite{vdH08}, Pan and Tsigaridas (2013a,b) \cite{PT13}, \cite{PTa},
 Kobel and Sagraloff (2013) \cite{KS13},
Pan (2015) \cite{P15}, and Pan (2015a) \cite{Pa}). 
This immediately implies  extension of  our algorithm that
support Corollary \ref{corref1} 
 to refining all simple isolated roots of a polynomial at a nearly optimal 
Boolean cost, but actually such an extension can be also obtained 
directly by using classical polynomial evaluation algorithm.

Various other iterative root-refiners
(see  McNamee (2002)
\cite{MN02}, McNamee (2007) \cite{MN07}, and
McNamee and Pan (2013) \cite{MNP13}) applied instead of Newton's
 also support our nearly optimal complexity estimates 
as long as   
isolation of the roots obtained by our power sum algorithm is
 sufficient in order to ensure subsequent superlinear 
convergence of the selected iterations.
In particular Ehrlich--Aberth's and WDK iterations converge globally
to all roots
with cubic and
quadratic rate, respectively, 
 if all the $d$ discs have isolation ratios at least
$3\sqrt d$, for Ehrlich--Aberth's iterations, and $8d/3$, 
for WDK iterations 
 (cf. \cite{T98}).


\begin{remark}\label{rempt}
The algorithm of \cite{P15} and \cite{Pa}
(the latter paper provides more details)
approximates
 a  polynomial of degree $d$ at $O(d)$ points
within $\epsilon z_+$, 
for $z_+$ of (\ref{eqz}) and
fixed $\epsilon$ such that $0<\epsilon<1$, 
by using  
$O(d\log^2 (d)(1+\log (1/\epsilon)))$ ops.
This matches the bound of 
\cite{MB72},
for $1/\epsilon$ of order $gd^h$ and
for positive constants $g$ and $h$. 
Moreover the algorithm  of \cite{P15} and \cite{Pa},
performed with the IEEE standard double precision,
 routinely
outputs close approximations to the values of 
the polynomial $p(x)$ of degree $d=4096$,
at $d$ selected points,
whereas the  algorithm of \cite{MB72} routinely fails
numerically, for $d$ of about 40 (cf. \cite{Pa}).

\end{remark}


\section{Boolean Cost Bounds}\label{sbcbr}
Hereafter \sOB  
denotes the
bit (Boolean) complexity ignoring logarithmic factors.
By $\lg(\cdot)$ we denote the logarithm with base 2.
To estimate the Boolean complexity of the algorithms supported by
Corollaries~\ref{corref} and \ref{corref1} we apply some results from
Pan and Tsigaridas (2013) \cite{PT13} and Pan and Tsigaridas (2015)
\cite{PTa}, which hold in the general case where the coefficients of
the polynomials are known up to an arbitrary precision. In
this section we assume that the polynomial $p=p(x)$
has Gaussian (that is, complex integer) coefficients
 known exactly; the parameter $\lambda$, to be
specified in the sequel, should be considered as the working
precision. We assume that we perform the computations using fixed
point arithmetic.

At first we consider the algorithm of approximating one complex root, $z$, of a polynomial $p$
up to any desired precision $\ell$. By an approximation we mean absolute approximation, that is compute a $\tilde z$ such that $\abs{z - \tilde z} \leq 2^{-\ell}$.
We assume that the degree of $p$ is $d$ and that 
$\normi{p} \leq 2^{\tau}$.

Following the discussion that preceded Theorem~\ref{thisol}, 
 we compute the polynomial $p_q$
and then  apply two DFTs, for $p_q$ and $\bar p_q$, 
and the inverse DFT, for $p_q/\bar p_q$.

Assume that $p$ is given by 
its $\lambda$-approximation $\wt{p}$ 
such that $\lg\normi{p -\wt{p}} \leq -\lambda$.
Perform all the operations with $\wt{p}$ and keep track of the precision loss
to estimate the precision of computations required in order
to obtain the desired approximation.

We compute $p_q$  using $d$ additions.
This results in a polynomial such that
$$\lg \normi{p_q} \leq \tau + \lg(d)$$
 and
$$\lg(\normi{p_q - \wt{p}_q}) \leq -\lambda + \tau \lg(d) + 1/2\lg^2(d) + 1/2\lg(d)
=  O(-\lambda + \tau\log(d) + \log^2(d)).$$
Similar bounds hold for $p_q'$, that is,
$$\lg(\normi{p'_q}) \leq \tau + 2\lg(d)$$ and
$$\lg(\normi{p'_q - \wt{p'}_q}) \leq  -\lambda + \tau \lg(d) + 3/2\lg^2(d) + 1/2\lg(d)
=  O(-\lambda + \tau\log(d) + \log^2(d)).$$

The application of  DFT on $p_q'$ 
leads us to the following bounds, 
\[
\abs{p'_q(\omega^{i})} \leq  \tau + 2\lg(d) + \lg\lg(d) + 2 = O(\tau + \log(d)) 
\]
and
$$
\abs{p'_q(\omega^{i}) - \wt{{p'_q(\omega^{i})}}} \leq 
-\lambda + \tau\lg(2d) + 3/2\lg^2(d) + 5/2\lg(d) + \lg\lg(d) + 5 = O(-\lambda + \tau\log(d) + \log^2(d)),$$
for all $i$,
\cite[Lemma~16]{PTa}.
Similar bounds hold for $p_q(\omega^{i})$.

The divisions
$k_i = p_q(\omega^{i})/p'_q(\omega^{i})$ 
output
complex numbers such that 
$$\abs{k_i} = \abs{p_q(\omega^{i})/p'_q(\omega^{i})} \leq  \tau + 2\lg{d} + \lg\lg{d} + 2.$$
Define  their approximations $\tilde {k}_i$
such that 
\[
\lg(\abs{k_i - \wt{k}_i}) \leq -\lambda + \tau\lg(4d) + 3/2\lg^2{d} + 9/2\lg{d} + 2\lg{\lg{d}} + 11 
= O(-\lambda + \tau\log(d) + \log^2(d)).
\]

The final DFT produces numbers
such that the logarithms of their magnitudes 
are not greater than
$\tau + 2\lg{d} + 2\lg\lg{d} + 4$
and the logarithms of their approximation errors 
are at most
$ -\lambda + \tau\lg(8d) + 3/2\lg^2{d} + 13/2\lg{d} + 4\lg{\lg{d}} + 18 = O(-\lambda + \tau\log(d) + \log^2(d))$,
\cite[Lemma~16]{PTa}.

To achieve an error 
within  $2^{-\ell}$ in the final result, we 
perform all the computations with accuracy
$\lambda = \ell +  \tau\lg(8d) + 3/2\lg^2{d} + 13/2\lg{d} + 4\lg{\lg{d}} + 18$,
that is 
$\lambda = O(\ell + \tau \log{d} + \log^{2}{d}) = \sO(\ell + \tau)$.

We perform $d$ additions 
at the cost $\OB(d  \, \lambda)$
and perform
the rest of computations, that is the 3 DFTs,
at the cost
$\OB( \log(d) \, \log\log(d) \, \mu(\lambda))$
or $\sOB(d (\ell + \tau))$ \cite[Lemma~16]{PTa}.
If the root that we want to refine is not in the unit disc, then we
replace $\tau$ in our bounds with $d \tau$.

We apply a similar analysis 
from
\cite[Section~2.3]{PT13} 
to the Newton 
iteration
 (see also
\cite[Section~2.3]{PTa}) and 
arrive at  the same 
asymptotic 
bounds on the Boolean complexity. 

In \cite{PT13} and \cite{PTa} 
the error bounds of Newton operator have been estimated 
by using the properties of real interval arithmetic.
In this paper we perform our computation in the field of 
 complex numbers, but this affects only the constants 
of interval arithmetic, and so
asymptotically, both the error bounds
and the complexity bounds of the Newton iterations are the same.
Thus, the overall complexity is $\sOB(d^2 \tau + d \ell)$
and the working precision is $O(d\tau + \ell)$.

In our case we also assume the  exact input, that is, 
assume the coefficients of the input
polynomials known up to arbitrary precision;
for example, they are integers.
For the refinement of 
the 
root
up to the precision of $L$ bits, we arrive at
an algorithm that supports the following  complexity estimates.
\begin{theorem}
  \label{thm:cref-one}
  Under the assumptions of Theorem~\ref{thisol} we can approximate
  the root $z$ of the polynomial $p(x) \in \mathbb{Z}[x]$, which is of degree $d$ and $\normi{p} \leq 2^{\tau}$,
  up to precision of $L$ bits in $\sOB(d^2 \tau + d L)$.
\end{theorem}

If we are interested in refining all complex roots, we cannot work
anymore with the polynomial $p_q$ of degree $q = O(\lg{d})$ unless
we add the cost of $d$ shifts of the initial approximations to the
origin.  Instead we rely on fast algorithms for multipoint
evaluation.  
Initially we evaluate the polynomial $p$ of degree $d$ at
$O(d \lg{d})$ points, and we assume that $\lg{\normi{p}} \leq \tau$.
These $d$ points approximate the roots of $p$, and so their magnitude
is at most $\leq 2^{\tau}$.

We use the following result of \cite[Lemma 21]{PTa}.
Similar bounds appear in \cite{K98,KS13,vdH08}.
\begin{theorem}[Modular representation]
  \label{lem:rem-m-polys}
  Assume that we are given $m+1$ polynomials,
 $F \in \mathbb{C}[x]$ of degree $2 m n$ 
  and $P_j \in \mathbb{C}[x]$ of degree $n$, for $j=1,\dots, m$
  such that $\normi{F} \leq 2^{\tau_1}$ and all roots of 
  the polynomials $P_j$
  for all $j$ have magnitude of at most
  $2^{\rho}$. 
  Furthermore assume $\lambda$-approximations of $F$ by  $\wt{F}$ and   
  of $P_j$ by $\wt{P}_j$
  such that 
  $\normi{F - \wt{F}} \leq 2^{-\lambda}$
  and 
  $\normi{P_j - \wt{P}_j} \leq 2^{-\lambda}$.
  Let  $ \ell = \lambda - O(\tau_1\lg{m} + m \, n \, \rho)$.
  Then we can compute 
  an $\ell$-approximations $\wt{F}_j$ of $F_j = F \mod P_j,$
  for $j=1,\dots, m$,
  such that $\normi{F_j - \wt{F}_j} \leq 2^{-\ell}$
  in 
  $\sOB(m \, n \,(\ell + \tau_1 + m\,n\,\rho))$.
\end{theorem}

Using this theorem we bound the overall complexity of multipoint
evaluation by $\sOB(d(L + d \tau))$.  The same bound holds at the
stage where we perform Newton's iteration.  We need to apply Newton's
operator $\sO(1)$ for each root.  Each application of the operators
consists of two polynomial evaluations.  We perform the evaluations
simultaneously and apply Theorem \ref{lem:rem-m-polys} to bound the
complexity.  On similar estimates for the refinement of the real roots
see \cite{PTa}.

We have the following theorem, which complements
Corollary~\ref{corref1} 
with the Boolean complexity estimates.

\begin{theorem}\label{thm:cref-all}
  Suppose that we are given $d$ discs, each containing a single simple root 
  of a polynomial $p(x) \in \mathbb{Z}[x]$ of degree $d$ and $\normi{p} \leq 2^{\tau}$, 
  and each being $(1+\eta)^2$-isolated, 
  for a fixed $\eta>0$.
  Then we can approximate all $d$ roots  of this polynomial 
  within $L$ bits in 
  in $\sOB(d^2 \tau + d L)$.
\end{theorem} 




\section{Acceleration of Henrici-Renegar's Winding Number Algorithm}\label{smhr}

Assume that $D(X, r):=\{x:~|X-x|\le r\}$ denote the disk of radius $r>0$ with a center $X$
and that no roots of a polynomial $p(x)$ lie  on its boundary 
circle  $C(X, r)=\partial D(X, r)=\{x:~|X-x|=r\}$. Let 
$\omega_j=X + re^{\sqrt{-1}2\pi j/d'}$, and  so
$\omega_0,...,\omega_{d^n-1}$ denote the $d'=2^{\lceil\log_2 16n\rceil}$ equally spaced points 
 on this   circle.

{\bf Step 1.} Evaluate $p(x)$ at $\omega_0,...,\omega_{d^n-1}$.

{\bf Step 2.} For each $i$, "label" $p(\omega_i)$ according to the scheme
\begin{equation}
L[p(x)] = \left\{
\begin{array}{cc}
1 & \textit{if }Re[p(x)]>0 \textit{ and }Im[p(x)]\ge 0\\
2 & \textit{if }Re[p(x)]\le 0 \textit{ and }Im[p(x)]> 0\\
3 & \textit{if }Re[p(x)]<0 \textit{ and }Im[p(x)]\le 0\\
4 & \textit{if }Re[p(x)]\ge 0 \textit{ and }Im[p(x)]< 0,\\
\end{array}
\right.
\end{equation}

where $Re[p(x)]$ and $Im[p(x)]$ denote the real and imaginary parts of $p(x)$.

{\bf Step 3.} If for some $i$, $L[p(\omega_{i+1})] - L[p(\omega_i)] = 2\mod 4$, then terminate and write "FAILED" (using the definition $\omega_{d'} = \omega_0$).

{\bf Step 4.} Defining $*:(a,b)\mapsto \{-1,0,1\}$ to be the operation on ordered pairs of integers $(a,b)$, $a-b\neq 2\mod 4$, where $a*b$ is congruent to $(a-b)\mod 4$, compute the integer

\begin{align*}
\# = \frac{1}{4}\sum_{i=0}^{d'-1}L[p(\omega_{i+1})] - L[p(\omega_{i})].
\end{align*}

Write "SUCCESS; $\#$."

\hfill

This algorithm returns the winding number of a polynomial $p(x)$ 
provided that no its roots are close to the boundary of the disk $D(X, r)$.

\begin{lemma}[\cite{R87}]
\label{Ren1}
Assume that all roots of a a polynomial $p(x)$ contained in $D(X, 3r/2)$ are also  contained in $D(X, r/6)$. 

Then the  above algorithm, applied to $D(X, r)$, returns "SUCCESS; $\#$" and $\#$ is the number of roots of $p(x)$ in $D(X, r)$.
\end{lemma}

By using only the signs of the real and imaginary parts of polynomial evaluations at points 
$\omega_0,...,\omega_{d^n-1}$, Henrici-Renegar's algorithm 
tracks the variation of  the polynomial evaluation $p(x)$ along the boundary of the disk $D(X, r)$. 
but there are still 
 numerical problems that must be resolved.

(i) If the real or imaginary part of an evaluation $p(\omega_j)$ is vanishing or is very close to 0
(even if $|p(\omega_j)|$ is not close to 0),
then the signs of the values $p(\omega_j)$ for some $j$
can be difficult to determine. 

(ii) The number computed by the algorithm may not be the correct winding number of the curve. 

Inspired by \cite{H74} and \cite{Z12}, we provide a remedy to these problems by 
 "tracking the location", namely, we further divide the complex plane, and use the
additional  information in  order to guarantee a stable output in spite
of potential numerical errors. For the ease of application, we also changed the region from a disk to the region bounded by the
square $S(z_0, r)$ with four vertices $\{z_0+r+r\sqrt {-1},z_0+r-r\sqrt {-1},z_0-r+r\sqrt {-1},z_0-r-r\sqrt {-1}\}$.

\hfill

{\bf Modified Renegar's Winding Number Algorithm}

Let $p(z)$ be a polynomial of degree $n$. Let $a_i = \frac{k}{n'}$ for $i = 0,...,n'$, $n'=2^{\lfloor\log_2 (64n/\pi)\rfloor}$. Let $L$ be an upper
bound of the magnitude of the derivative $\frac{dp(\gamma(t))}{dt}|$ over $[0,1]$.

{\bf Step 1.} (Check for singularity.) Evaluate $p(z)$ at $a_0, ..., a_{n'-1}$. If any of these evaluation has absolute value less than $\sqrt{2}(r/2)^n$, then terminate and return "FAILED". 

{\bf Step 2.} (Check for correctness of the computations.) Check if the sum of absolute values of two consecutive 
evaluations is less than $(8rL)/n'$. If any such a sum is less than this bound, terminate and return "FAILED".

{\bf Step 3.} For each $i$, label $p(a_i)$ according to its argument:

\begin{equation}
L[p(x)] = \lfloor \frac{4arg(p(x))}{\pi} \rfloor
\end{equation}

Since $0 \leq arg(p(x)) < 2\pi$, $L[f(x)]$ returns an integer between 0 and $7$. In other words, $L[p(x)]$ marks the octant to which the evaluation belongs.


{\bf Step 4.} Define an operation on ordered pairs of integers:  $a*b := (a-b)(mod \text{ 8})$. We  show in Lemma \ref{Lem1} that if there is no root of $p(z)$ close to the boundary of the disk $D(z_0, r)$, then the product $L[p(a_{i+1})] * L[p(a_i)]$ can only take a value in 
the set $\{-1, 0, 1\}$. Lemma \ref{Lem3} shows that the algorithm works even if each labelling is off by $\pm 1$.

{\bf Step 5.} Compute and return the integer
\begin{equation}
\# = \frac{1}{8}\sum_{i=0}^{n'-1}L[p(a_{i+1})] * L[p(a_i)].
\end{equation}

Then $\#$ denotes the winding number of $p(z)$ on the boundary of the square $S(z_0, r)$.

\begin{lemma}
\label{Lem1}
Let $\alpha = 1/2+\sqrt{2}$ and $\beta = 1/2$. Assume that the only roots of $p(z)$ contained in 
the disc $D(z_0, \alpha r)$ are contained in the disc  $D(z_0, \beta r)$. Then the operation $a*b$ in the algorithm 
above can only produce a value in the set $\{-1, 0, 1\}$.
\end{lemma}

\begin{proof}
Define a parametrization of the square $S(z_0, r)$ as follows:
\begin{align*}
\gamma(t)&: [0,1]\rightarrow S(z_0, r)\\
\gamma(t) =& \left\{
\begin{array}{cc}
z_0+r+r\sqrt{-1}-8rt & 0 \le t < \frac{1}{4}\\
z_0+r-r\sqrt{-1}-(8t-2)r\sqrt{-1} & \frac{1}{4} \le t < \frac{1}{2}\\
z_0-r-r\sqrt{-1}+(8t-4)r & \frac{1}{2} \le t < \frac{3}{4}\\
z_0+r-r\sqrt{-1}+(8t-6)r\sqrt{-1} & \frac{1}{4} \le t \le 1
\end{array}
\right.
\end{align*}
Simply put, $\gamma(t)$ travels along $S(z_0, r)$ once at constant pace, such that 
$$|\gamma(a_{i+1})-\gamma(a_i)| = |a_{i+1} - a_i| = 8r/n', i=0, ..., n'-1.$$
Next we  show that under the hypothesis of the lemma, the argument of $p(\gamma(t))$ cannot vary by more than $\pi/4$
 as $t$ increases from $a_i$ to $a_{i+1}$. Therefore $L[p(\omega_{i+1})] * L[p(\omega_i)]$ can only take a value in the set $\{-1, 0, 1\}$.

Suppose otherwise, then for some $t_i \leq t' < t'' \leq t_{i+1}$, the argument of $p(\gamma(t'))$ would differ
from  that of $p(\gamma(t''))$ by $\pi/4$. Let $\xi_1, ..., \xi_n$ be the roots of $p(z)$. Since $arg(p(z)) = \sum_i arg(z - \xi_i) (mod 2\pi)$, for at least one root $\xi$, the argument of $\gamma(t') - \xi$ would differ from that of $\gamma(t'') - \xi$ by more than $\pi/4n$. Now let
us show that this is impossible. Let \textit{dis} denote the distance from $\xi$ to the line segment connecting $\gamma(t')$ and $\gamma(t'')$. Under the hypothesis of the lemma, we have $dis \ge r/2$,

\begin{align*}
&|arg(\gamma(t') - \xi) - arg(\gamma(t'') - \xi)| \\
=&  2 tan^{-1} (\frac{|\gamma(t') - \gamma(t'')|}{2\cdot dis})\\
=&  2 tan^{-1} (\frac{|t' - t''|}{2\cdot dis})\\
\\
	\le&  2 tan^{-1}(\frac{8r}{2n'\cdot dis})
\\
	\le&  2 tan^{-1}(\frac{8}{n'})
\\
	\le&  2 tan^{-1}(\frac{\pi}{8n})
\\
	<&  \frac{\pi}{4n}
\end{align*}

Therefore the argument of $p(\gamma(t))$ cannot vary by more than $\pi/4$
as $t$ increases from $t_i$ to $t_{i+1}$. This prove the lemma.
\end{proof}

As explained in the algorithm description, $L[p(\omega_{i})]$ marks the octant to which $p(\omega_{i})$ belongs, and the value $L[p(\omega_{i+1})] * L[p(\omega_i)]$ indicates the change of octant as $i$ runs through $1, ..., n'$. The value is 1 if the evaluation rotates counter-clockwise, -1 if the evaluation rotates clockwise, and 0 otherwise. Thus every full counter-clockwise cycle is indicated by an increase of $\#$ by 1. The following lemma guarantees that when the evaluations are not small, the algorithm misses no turn of the curve.

\begin{lemma}
\label{Lem_cor}
Let $L$ be an upper bound of the magnitude of the derivative $\frac{dp(\gamma(t))}{dt}|$ over $[0,1]$. If $|p(\gamma(a_i))| > \frac{8rL}{n'}$ for all $i=0,...,n'$, then the integer returns the correct winding number of $p(z)$ around the square $S(z_0, r)$.
\end{lemma}
\begin{proof}
This lemma follows directly from [Lemma 3]\cite{Z12}.
\end{proof}

{\bf Remark.} We can decrease the gap between $D(z_0, \alpha r)$ and $D(z_0, \beta r)$ by increasing $n'$. In fact, it can be seen from the proof that the gap can be decreased by half every time we double the value of $n'$. From now on we  call the algorithm with number of samples doubled $g$-times as the modified winding number algorithm of degree $g$. 

Let us next estimate precision required in the above computations. What precision is enough for labelling the evaluations correctly? For the precise labelling we need to know the signs of both real and imaginary part of the evaluations, as well as the correct comparison of its values. This requirement is impractical as these values can be less than any working precision. However, the following discussion indicates that even if the labels are off by $\pm 1$, the algorithm still correctly returns the winding number.

\begin{definition}
For each $k = 0, ..., n'$, let {\bf $L_k$} denote the true label of the evaluation $p(a_k)$, and let {\bf $\widetilde{L_k}$} denote the actual label of the evaluation. Note that $L_{n'} = L_0$ and $\widetilde{L_{n'}} = \widetilde{L_0}$.
\end{definition}

The following discussion is based on the condition of Lemma \ref{Lem1} and the following assumption:

{\bf Assumption.} 
For each $k = 0, ..., n'$, the computation of its real and imaginary part is precise enough so that $|\widetilde{L_k} - L_k| \leq 1$.

The assumption is satisfied if there is no complex roots near the square $S(z_0, r)$
because  the absolute value of such an evaluation can be readily bounded. If neither real nor imaginary part is small enough to create confusion, then at least $\widetilde{L_k}$ points to the same quadrant as $L_k$. If one of them is indeed small, 
then the other one must be large enough so that the comparison between the real and the imaginary part is also 
clear. In particular, we can prove the following result:
\begin{lemma}
Keep the assumptions of Lemma \ref{Lem1}. If the working precision $2^{-b}$ is greater than $\sqrt{2}(r/2)^n$, 
then for any $z\in S(z_0, r)$, at most one of the following events  occurs:
\begin{itemize}
\item
$|Re(p(z))|<2^{-b}$
\item
$|Im(p(z))|<2^{-b}$
\item
$|Re(p(z)) + Im(p(z))|<2^{-b}$
\item
$|Re(p(z)) - Im(p(z))|<2^{-b}$
\end{itemize}
\end{lemma}
\begin{proof}
For $z\in S(z_0, r)$, the norm of $p(z)$ can be bounded as follows:
\begin{align*}
&|p(z)|\\
=&|\prod_{j=1}^{n}(z-z_j)|\\
=&\prod_{j=1}^{n}|z-z_j|\\
\ge& \prod_{j=1}^{n} r/2\\
=& (r/2)^n\\
\ge& 2^{-b}/\sqrt{2}
\end{align*}
Considering $|p(z)|^2 = |Re(p(z))|^2 + |Im(p(z))|^2$, the occurrence of any two  of the four aforementioned conditions  results in $|p(z)|<2^{-b}$, thus proving the lemma.
\end{proof}

\begin{lemma}
\label{Lem2}
For each $k = 0, ..., n'$, $\widetilde{L_{k+1}} * \widetilde{L_k} \neq 4 (mod \text{ 8})$ 
\end{lemma}

\begin{proof}
Lemma \ref{Lem1} essentially proves that  $L_{k+1} * L_k = 0,\pm1 (mod \text{ 8})$. Since $\widetilde{L_k}$ can only differ from $L_k$ by at most one, it can be easily seen that $\widetilde{L_{k+1}} * \widetilde{L_k} \neq 4 (mod \text{ 8})$.
\end{proof}

Let $\epsilon_k = L_k * \widetilde{L_k}\text{ (mod 8)}$ denote the difference between the true label and the computed label. By assumption $\epsilon_k$ can be chosen such that $|\epsilon_k| \leq 1$. Lemma \ref{Lem3} shows that the computation of the winding number is not be affected by such errors.

\begin{lemma}
\label{Lem3}
$$
\# = \frac{1}{8}\sum_{k=0}^{n'-1}[\widetilde{L_{k+1}} * \widetilde{L_k}].
$$
\end{lemma}

\begin{proof}
For each $k = 0, ..., n'$,

\begin{align*}
&(L_{k+1} * L_k) - (\widetilde{L_{k+1}} * \widetilde{L_k})\\
& \equiv (L_{k+1} - L_k) - (\widetilde{L_{k+1}} - \widetilde{L_k}) \text{ (mod 8)}
\\
& = (L_{k+1} - \widetilde{L_{k+1}}) - (L_k - \widetilde{L_k}) \text{ (mod 8)}
\\
& \equiv (L_{k+1} * \widetilde{L_{k+1}}) - (L_k - \widetilde{L_k}) \text{ (mod 8)}
\\
& \equiv \epsilon_{k+1} - \epsilon_{k} \text{ (mod 8)}
\end{align*}

Since both $(L_{k+1} * L_k) - (\widetilde{L_{k+1}} * \widetilde{L_k})$ and $\epsilon_{k+1} - \epsilon_{k}$ belongs to $\{0,\pm1,\pm2\}$, we must have
\begin{equation}
(L_{k+1} * L_k) - (\widetilde{L_{k+1}} * \widetilde{L_k}) = \epsilon_{k+1} - \epsilon_{k}, \forall k = 1,...,n'.
\end{equation}
Therefore
\begin{align*}
&\# - \frac{1}{8}\sum_{k=0}^{n'-1}[\widetilde{L_{k+1}} * \widetilde{L_k}]\\
& = \frac{1}{8}\sum_{k=0}^{n'-1}[(L_{k+1} * L_k) - (\widetilde{L_{k+1}} * \widetilde{L_k})]
\\
& = \frac{1}{8}\sum_{k=0}^{n'-1}[\epsilon_{k+1} - \epsilon_k]
\\
& = \epsilon_{n'} - \epsilon_0
\\
& = 0.
\end{align*}

This completes the proof of  Theorem \ref{Lem3}.
\end{proof}

The following theorem  summarizes our results:
\begin{theorem}
\label{thm1}
Let $\alpha = (1/2)^g+\sqrt{2}$ and $\beta = 1-(1/2)^g$. Assume that the only roots of $p(z)$ contained in $D(z_0, \alpha r)$ are contained in $D(z_0, \beta r)$
as well. Further assume that the working precision $2^{-b}$ is smaller than $\min\{\sqrt{2}(r/2)^n, (8rL)/2^{g\lfloor\log_2 (64n/\pi)\rfloor}\}$.
Then the modified winding number algorithm of degree $g$ will return the number of roots of $p(z)$ in $S(z_0, r)$. 
\end{theorem}

Given the square $S = S(z_0, r)$, let $\tau_p$ denote the minimum distance between a root $p(x)$ and this square.
Then, by applying the modified winding number algorithm of degree $g=O(\log(1/\tau_p))$ o $S$,
we correctly calculate the number of roots in the region bounded by $S$. This requires $n' = 2^{g\lfloor\log_2 (64n/\pi)\rfloor}$ polynomial evaluations 
with a precision up to $\min\{\sqrt{2}(r/2)^n, (8rL)/n'\}$. Suppose that the cost of single polynomial evaluation with 
the precision $2^{-b}$ is $E(n, b)$, then the required precision $b$ is asymptotically equal to $O(n\log L\log(1/r)+\log(1/\tau_p)\log n)$
 and the overall evaluation cost is $O(n(1/\tau_p)E(n, b))$.

\subsection{Summary}

Now we summarize the key properties of the modified winding number algorithm of degree k:
\begin{itemize}
\item
If the algorithm succeeds on a square $S = S(z_0, r)$, then it returns the number of roots in $S$.
\item
If there is no roots between $D(z_0, [(1/2)^g+\sqrt{2}]r)$ and $D(z_0, [1-(1/2)^g]r)$, then the algorithm is guaranteed to succeed.
\item
The algorithm requires polynomial evaluation  with $b = O(n\log L\log(1/r)+\log(1/\tau_p)\log n)$-bit of precision,
 and the bit-complexity is $O(n(1/\tau_p)E(n, b))$, where $E(n,b)$ denotes the cost of one $b$-bit polynomial evaluation.
\end{itemize}



\section{Conclusions}\label{sconcl}


Our polynomial root-finding consists of
three stages.

(i) At first, one computes some initial isolation 
of the roots or the clusters of the roots
 of a polynomial of a degree $d$
 by applying some known algorithms.
One can apply Ehrlich--Aberth's and  WDK iterations, which
have excellent empirical record of fast and reliable global convergence,
but this record has not been supported by formal analysis, 
and the iterations are not efficient for the approximation
of roots isolated in a disc or in a  fixed region on the complex plane. 
One can avoid these deficiencies by applying advanced variants of 
Weyl's Quad-tree, a.k.a. subdivision, algorithms. In particular
a variant of \cite{P00} yields such an initial isolation
at a nearly optimal arithmetic cost, and furthermore control
of the precision of computing is rather straightforward at
most of its steps.  By incorporating our new accelerated version  of 
winding number algorithm by Henrici--Renegar, we enable
application  of the entire  construction 
within a nearly optimal cost bound even where a polyomial 
is given by a black box subroutine for its evaluation. 

(ii) Next one increases the isolation in order to
facilitate the final stage.

(iii) Finally,  
one computes close approximations to all roots 
or the roots in a fixed region
by applying 
a simplified version of
the algorithm of \cite{K98} presented in 
\cite{P12} and \cite[Section 15.23]{MNP13}.
For approximations to all roots one can
instead apply
  Ehrlich--Aberth's or  WDK  iterations. 

Besides acceleration of winding numbe algorithm, we contribute at stage (ii). 
Namely, we measure the isolation of a root or a root cluster
by the isolation ratio of its covering disc,
that is, by the ratio  of the distance from the center
of the disc to the external
roots and the disc radius. Then our algorithm 
 increases the isolation ratio 
from any constant $1+\eta$, for a positive $\eta$,
to $gd^h$, for any pair of positive constants $g$ and $h$.
The algorithm is performed at nearly optimal arithmetic and 
Boolean cost and can be applied
even when a polynomial is given by a black box for its evaluation, 
while its coefficients  are unknown.

If such an isolation is ensured, then 
Newton's iterations  of 
 \cite{K98},
\cite{P12} and \cite[Section 15.23]{MNP13} 
as well as  Ehrlich--Aberth's and  WDK  iterations
  converge 
 right from this point with quadratic 
or cubic rate (cf. \cite{T98}) and also perform at nearly optimal
arithmetic and Boolean cost even where a polynomial is given by 
 a black box for its evaluation.

Our analysis can be readily extended 
to prove the same results even if we  assume
weaker initial isolation with the ratio as low as 
$1+\eta$, for $\eta$ of order $1/\log (n)$.
Then we could move to
 stage (ii)  after fewer iterations of stage (i).
Further extension of our 
results to the case of $\eta$ of order $1/n^h$,
for a positive constant $h$, is a theoretical challenge, but
 progress at stages (ii) and (iii) 
is  important for both theory and practice.

Unlike Ehrlich--Aberth's and  WDK  iterations,
Newton's iterations  of 
 \cite{K98},
\cite{P12} and \cite[Section 15.23]{MNP13} 
produce (as by-product) polynomial factorization,
which  is important on its own right
(see the first paragraph of the introduction)
and enable root-finding in
the cases of multiple and clustered roots.
Let us conclude with some comments on this subject.

Our Theorem \ref{thisol} covers the case 
where we are given a $(1+\eta)$-isolated disc 
containing $k$ roots of a polynomial $p(x)$ of (\ref{eqpoly}).
In that case, at the cost of
 performing $O(d+\log(d)\log(\log (d)))$ ops,
 Algorithm \ref{algpwrs}
outputs approximations to the power sums of these roots
within the error bound $\Delta_{q,k}$ of order 
$r^{q-k}$, for $r=1/(1+\eta)$ and positive integers $q=O(d)$ and $k<q$
of our choice, say, $q=2k$.
The algorithm can be extended to the case where
all roots or root clusters of the polynomial $p(x)$ 
are covered by $s$ such $(1+\eta)$-isolated discs. 
Then we can still approximate the
power sums of the roots in all discs
simultaneously by using $O(d\log^2(d))$ ops.

At this point the known numerically  stable algorithm for
the transition from the power  sums to the coefficients
(cf. \cite[Problem I$\cdot$POWER$\cdot$SUMS, pages 34--35]{BP94})
approximates the coefficients of the $s$ factors of 
$p(x)$ whose root sets are precisely the roots of 
$p(x)$ lying in these $s$ discs.
The algorithm performs within 
dominated arithmetic and Boolean cost bounds.

Having these factors computed, 
one can reduce root-finding for $p(x)$
to root-finding for the $s$ factors
of smaller degrees. 
In particular, this would cover root-finding in the
harder case where the discs 
contain multiple roots or root clusters.
As this has been estimated in \cite{MNP12}
and \cite{P12}, such a divide and conquer approach 
can dramatically accelerate stage (iii),
even versus the Ehrich--Aberth's and WDK algorithms,
except that the initialization of the respective
algorithm requires a little closer approximation 
to the coefficients of the $s$ factors of $p(x)$
than we obtain from the power sums of their roots.

The algorithm of \cite{K98}, extending the one of 
\cite{schoenhage82}, produces such a desired refinement
of the output of Algorithm \ref{algpwrs}
at a sufficiently low asymptotic Boolean cost,
but is quite involved.
In particular, it
reduces all operations with polynomials
to multiplication of long integers.
In our further work we will seek the same asymptotic complexity results
at this stage,
 but with smaller overhead constants,  based
on the simplified version  
 \cite{P12}
and  \cite[Section 15.23]{MNP13}
of the algorithm of \cite{K98}. 


\paragraph*{Acknowledgments.}
 VP has been supported by NSF Grant CCF
1116736 and by PSC CUNY Award 67699-00 45.  ET has been partially
supported by GeoLMI (ANR 2011 BS03 011 06), HPAC (ANR ANR-11-BS02-013)
and an FP7 Marie Curie Career Integration Grant.




\end{document}